\documentclass{lmcs}    

\title{The Complexity of Unavoidable Word Patterns}

\author{Paul Sauer}
\address{University of South Africa}
\email{paul.v.sauer@gmail.com}


\theoremstyle{definition}
\newcommand{\href}{}
\newcommand{\Sn}{S_n}

\newcommand{\la}{\langle}
\newcommand{\ra}{\rangle}
\newcommand{\pipat}{\la\pi\ra}

\newcommand{\rN}{[r]^N}
\newcommand{\N}{\mathbb{N}}

\newcommand{\al}{\alpha}
\newcommand{\alp}{\al_{\phi}}
\newcommand{\Ga}{G_{\alpha}}
\newcommand{\ov}{\overline}
\newcommand{\zp}{z_+}
\newcommand{\A}{A_{\al}}

\begin{document}

\begin{abstract}    
The avoidability, or unavoidability of patterns in words over finite alphabets has been studied extensively. The word $\alpha$ over a finite set $A$ is said to be unavoidable for an infinite set $B^+$ of nonempty words over a finite set $B$ if, for all but finitely many elements $w$ of $B^+$, there exists a semigroup morphism $\phi:A^+\rightarrow B^+$ such that $\phi(\alpha)$ is a factor of $w$. We present various complexity-related properties of unavoidable words. For words that are unavoidable, we provide an upper bound to the
lengths of words that avoid them. In particular, for a pattern $\alpha$ of length $n$ over an alphabet of size $r$, we give a concrete function $N(n,r)$ such that no word of length $N(n,r)$ over the alphabet of size $r$ avoids $\alpha$.

A natural subsequent question is how many unavoidable words there are. We show that the fraction of words that are unavoidable drops exponentially fast in the length of the word. This allows us to calculate an upper bound on the number of unavoidable patterns for any given finite alphabet.

Subsequently, we investigate computational aspects of unavoidable words. In particular, we exhibit concrete algorithms for determining whether a word is unavoidable. We also prove results on the computational complexity of the problem of determining whether a given word is unavoidable.

\end{abstract}

\maketitle

\section{Introduction}

Let $\N$ denote the nonnegative integers. If $A$ is a finite set, we write $A^*$ for the set $\{a_1a_2\dots a_n \mbox{ }  | \mbox{ } a_i \in A \mbox{ and }  n\in\N\}$ of {\sl words} over $A$, while $A^+$ is the subset of all nonempty words in $A^*$. 
  For $n\in\N$ we symbolize the set of words of length $n$ over $A$ by $A^n$. Here the {\sl length} of a word is defined in the conventional sense: if $w\in A^*$ and $w=a_1a_2\dots a_n$ with each $a_i\in A$, then the length $|w|$ of $w$ is $n$.  The set $A$ above is sometimes called an {\sl alphabet} and its members are called {\sl letters}. We say that the word $v = a_1a_2\dots a_m$ is a {\sl factor} of the word $w=b_1b_2\dots b_n$ if there is an $i$ such that, for $1 \le j \le m$, we have $a_j = b_{{i_0}+j}$.

For a word $w$ and letters $x_1, x_2, \dots, x_k$, we denote by $w^{x_1, x_2, \dots, x_k}$ the word derived from $w$ by deleting all occurrences of each of the $x_i$.

We say that a word $w$ over a finite alphabet $B$ {\sl reflects} a word $\alpha$ (or a {\sl pattern} $\alpha$, for the sake of clarity)  over a finite alphabet $A$ whenever there is a semigroup morphism $\phi:A^+ \rightarrow B^+$ such that $\phi(\alpha)$ is a factor of $w$. The pattern $\alpha$ is called {\sl unavoidable} for a set $X$ of words over a finite alphabet if all but finitely many $w\in X$ reflect $\al$. The pattern $\al$ is simply called unavoidable if the preceding statement holds for every set over a finite alphabet. Otherwise $\al$ is called {\sl avoidable}.

The study of combinatorial patterns is one of the most repeated themes in Mathematics \cite{devlin1996mathematics}, \cite{flo2009}. Among these studies, the unavoidability of patterns in words over finite alphabets has been explored extensively. Over the last century, this theme has resurfaced repeatedly \cite{thue1906},  \cite{morse1944}, \cite{bean1979}, \cite{zim84}, \cite{roth92}, \cite{Schmidt1987}. In the last decade, there has been a resurgence in the investigation of unavoidability \cite{sapir2016}. Thue \cite{thue1906} proved that $xxx$ is avoidable on the binary alphabet and $xx$ is avoidable on the alphabet of size $3$. Bean et al. \cite{bean1979} conducted an extensive investigation into the avoidability of patterns. One central discovery of this investigation is the notion of a letter that is {\sl free} for a pattern.

\newtheorem{freedef}[thm]{Definition}
\begin{freedef}\label{thm:freedef}
Let A be a finite alphabet and let $w\in A^+$. A letter $x\in A$ is free for $w$ if $x$ occurs in $w$ and there is no integer $n > 0$ and $a_1, a_2,\dots,a_n,b_1,b_2,\dots,b_n$ such that 
 \begin{align*}
 x&a_1\\
 b_1&a_1\\
 b_1&a_2\\
 b_2&a_2\\
 &\vdots\\
 b_n&x\\
 \end{align*}
 are all factors of $w$.
\end{freedef} 

Free letters are connected to the phenomenon of unavoidability be the following lemma, whose proof appears in \cite{bean1979}.

\newtheorem{bemfree}[thm]{Lemma}
\begin{bemfree}\label{thm:bemfree}
Suppose $\alpha$ is a pattern with a free letter $x$. If $\alpha^x$ is unavoidable, then so $\alpha$.
\end{bemfree}

A surprising, complete characterization of unavoidable patterns follows from Lemma \ref{thm:bemfree}. This is commonly known as the Bean, \index{Ehrenfeucht}Ehrenfeucht and \index{McNulty} McNulty (B.E.M.) Theorem.

\newtheorem{bemmain}[thm]{B.E.M. Theorem}
\begin{bemmain}\label{thm:bemmain}
A pattern $\alpha$ is unavoidable if and only if it is reducible to the empty word by iteratively performing one of the following operations on the pattern: 
\begin{enumerate}
\item deleting every occurrence of a free letter, or
\item replacing all occurrences of some letter $x$ occurring in $\alpha$ by a different letter $y$, also occurring in $\alpha$. 
\end{enumerate}
\end{bemmain}

We refer to the second operation as the {\sl identification of letters}. We can extend the definition of free letters to {\sl free sets}.
\newtheorem{freesetdef}[thm]{Definition}
\begin{freesetdef}\label{thm:freesetdef}
Let A be a finite alphabet and let $w\in A^+$. A set $X\subseteq A$ is free for $w$ if, for every pair of letters $x,y \in X$, there is no $n > 0$ and $a_1, a_2,\dots,a_n,b_1,b_2,\dots,b_n$ such that 
 \begin{align*}
 x&a_1\\
 b_1&a_1\\
 b_1&a_2\\
 b_2&a_2\\
 &\vdots\\
 b_n&y\\
 \end{align*}
 are all factors of $w$.
\end{freesetdef} 

The notion of free sets allows us to reformulate the B.E.M. Theorem in a way that is sometimes more convenient for reasoning about patterns. 
\newtheorem{bembetter}[thm]{Theorem}
\begin{bembetter}\label{thm:bembetter}
A pattern is unavoidable if and only if it is reducible to the empty word by iteratively deleting free sets.
\end{bembetter}
This reformulation is due to Sapir \cite{sap87}. See also Zimin \cite{zim84}.
The proof of Theorem \ref{thm:bemmain}, presented in \cite{bean1979}, is not constructive. Therefore it gives no indication, for any given pattern, what the longest word avoiding that pattern might be. Subsequent to \cite{bean1979}, one constructive unavoidability result was established, pertaining to the subset of patterns that represent permutations. We now discuss this result briefly.

Let $[n]$ denote the set $\{1,2,\dots, n\}$ and let $\Sn$ be the set of all permutations of $[n]$. We use one-line notation to express a permutation $\pi\in\Sn$ -- that is we write $x_1 x_2 \dots x_n$ when $\pi(i) = x_i$ for $i\in[n]$. The write $\pipat$ for the word $12\dots nx_0x_1\dots x_n$, where $x_0$ is a symbol not in $[n]$. Fouch\'e \cite{fou95} discovered the following
\newtheorem{fouche}[thm]{Theorem}
\begin{fouche}
For $n,r\in\N$ there is an $N = N(n,r)\in\N$ such that every $w\in\rN$ reflects every $\pipat$, where $\pi\in\Sn$. 
Specifically, the numbers $N(n,r)$ are inductively bounded from above by 
\begin{equation*}
N(n+1, r+1) \le 2(n+1)N(n+1, r)N(n, (2n+2)^2 r^{N(n+1,r)})
\end{equation*}

\end{fouche}

In the sequel, we show that a similar bound holds for all unavoidable patterns. 
The proof of the Main Theorem \ref{thm:mainres} follows Fouch\'e's reasoning. Subsequent sections are organized as follows: 

In Section \ref{sdensity}, we investigate the density of unavoidable patterns in the space of all patterns. We establish that this density drops quite fast as the length of the pattern increases. This fact then provides a way to calculate an upper bound for the number of unavoidable patterns as function of the size of the underlying alphabet.

Section \ref{sfree} is devoted to the algorithmic decision problem of whether a letter appearing in a given pattern is free. We present a concrete algorithm running in polynomial time. In Section \ref{slogic}, we show that there is a simple reduction from boolean formulas to patterns that maps satisfiable formulas to unavoidable patterns and unsatisfiable formulas to avoidable patterns. The final substantial part of the paper is Section \ref{scomplexity}, where we prove that the the problem of deciding whether a pattern is unavoidable is $NP$-complete.

\section{General Bounds for Unavoidable Patterns}

The main result of this section is Theorem \ref{thm:mainres}, which provides an upper bound on the length of words that can avoid a given, unavoidable pattern. In order to establish Theorem \ref{thm:mainres}, we first need to establish a few facts. Lemma \ref{thm:morph} below gives us a method for building morphisms as the size of our alphabet increases, provided that there is a free letter in the pattern. This lemma, stated here without proof, is proved in \cite{bean1979}.

\newtheorem{morph}[thm]{Lemma}
\begin{morph}\label{thm:morph}
Let A and B be a finite alphabets and let  $w$ be a word over $A$. Suppose $x$ is free for $w$. If there is a morphism $\phi:w^x\mapsto v$, where $v\in B^+$ is of the form $a^{i_1}X_1a^{i_2}X_2\dots a^{i_t}X_ta^{i_{t+1}}$, each $X_i$ being a word over $B\setminus \{a\}$, then there is a morphism $\psi:w\mapsto v$. 
\end{morph}

Since every letter in a free set is free, the following Lemma follows immediately from Theorem \ref{thm:bembetter}. This will be used in conjunction with Lemma \ref{thm:morph} to build morphisms in the proof of the main result below.

\newtheorem{free}[thm]{Lemma}
\begin{free}\label{thm:free}
Every unavoidable pattern has a free letter.
\end{free}

We are now ready to prove our main result. The construction of the proof closely follows \cite{fou95}.

\newtheorem{mainres}[thm]{Main Theorem}
\begin{mainres}\label{thm:mainres}
For $n,r\in\N$ there is an $N = N(n,r)\in\N$ such that every $w\in\rN$ reflects every unavoidable pattern of length $n$ over $[r]$. The minimal values for the numbers $N(n,r)$ are bounded from above by 
\begin{equation*}
N(n+1, r+1) \le (n+1)N(n+1, r)N(n, (n+1)^2 r^{N(n+1,r)})
\end{equation*}
\end{mainres}

\begin{proof}
It is easy to see that $N(1,r) = r+1$ and $N(1,n) = n+1$. From here we proceed by induction to establish the stated bound. Suppose our result holds for some $n$ and all $r$, as well as for $n+1$ and some $r \ge 1$. 

Let $w$ be a word of length $(n+1)KL$ over an alphabet $A$ of size $r+1$, where $K=N(n+1, r)$ and $L=N(n, (n+1)^2r^{N(n+1, r)}$. We may assume that every factor of length $K$ in $w$ contains every letter in $A$, for otherwise $w$ reflects every unavoidable pattern of length $n+1$, by our inductive hypothesis. Consequently, the word $w$ is of the form $a^{i_1}X_1a^{i_2}X_2\dots a^{i_t}X_ta^{i_{t+1}}$, where each $X_i \in \{A\setminus \{a\}\}^{+}$ satisfies $|X_i| < K$. We may assume that  $1\le a_{i_j} \le n$, for otherwise the morphism $f(x) = a$ that sends every letter to $a$ shows that every pattern of length $n+1$ is reflected by $w$. 

We immediately have 
\begin{align*}
(n+1)KL = |w| &\le (K-1)t + (t+1)(n+1)\\
&= (K+n)t +n +1\\
&\le (n+1)Kt +1\\
\end{align*} 
since $K > 2$ is readily available from the definition of K. Therefore we have $t > L$ and hence $w$ has a factor $v = a^{i_1}X_1a^{i_2}X_2\dots a^{i_L}X_L^{i_{L+1}}$, where each $X_i \in \{A\setminus \{a\}\}^{+}$ satisfies $|X_i| < K$.

Define the alphabet $B$ as the set of words of the form $a^iX$, with $1\le i\le n$ and  $X \in \{A\setminus \{a\}\}^{+}$ satisfies $|X_i| < K$. 

\begin{align*}
 |B| =& n(r+r^2+\dots+r^{K-1})\\
 \le& (n+1)^2r^K\\
\end{align*}
since $K \ge n + 1$ for every $n$ and $r$.

We have $v$ is a word of length $L$ over $B$. Suppose that $\alpha$ is any unavoidable pattern of length $n+1$ over $A$. Using Lemma \ref{thm:free} there is a letter $x\in A$ that is free for $\alpha$. We remind ourselves that $L=N(n, (n+1)^2r^{N(n+1, r)})$ and note, by our inductive hypothesis, that there is thus a morphism $\phi:\alpha^x\mapsto v$. Consequently, Lemma \ref{thm:morph} yields that there is a morphism $\psi:\alpha\mapsto v$ and the proof is complete.
\end{proof}

%
%

%
%

\section{Density and Counting Unavoidable Patterns}\label{sdensity}
A natural subsequent question is how many unavoidable words there are. We start by showing that, for alphabets of $3$ or more letters, the fraction of words that are unavoidable drops exponentially fast in the length of the word.

\newtheorem{avdense}[thm]{Lemma}
\begin{avdense}\label{thm:avdense}
Let $r > 2$ and $n > 0$. Let $p_{r,n}$ be the probability that a pattern of length n is unavoidable over [r]. We have $p_{r,n} \le \left(\frac{r-1}{r}\right)^{n-1}$.
\end{avdense}

\begin{proof}
Let $w$ be a word of length $n$ over $r$. If $n=1$ then $w$ is unavoidable, so that our claim holds with $p_{r,1} = (\frac{r-1}{r})^0 = 1$. Now suppose $n > 1$. We will use the fact that $xx$ is avoidable, established in \cite{thue1906}. Let $V = \{w\in[r]^n :  x\in[r] \mbox{ and } xx \mbox{ is a factor of } w\}$. First we claim that every element of $V$ is avoidable. To prove our claim, we start by noting that $x$ is not free for any $v\in[r]^*$ that has $xx$ as a factor. Hence any sequence of deletions of free letters applied to $w$ results in a word that has $xx$ as a factor. Using Theorem \ref{thm:bemmain}, our claim is proved. Let $U_{n,r}$ be the set of all unavoidable words of length $n$ over $r$. By our claim above, we have $U \subseteq \bar{V} = [r]^n\setminus V$. Now we count the elements of $\bar{V}$. Let $w = w_1w_2\dots w_n$ be an abstract word of length $n$ over $r$. For $w_1$ we can choose any one of the $r$ letters in $[r]$. For each subsequent $w_i$, we can choose any letter from $[r]$, other than our choice of $w_{i-1}$. Hence $|\bar{V}| = r(r-1)^{n-1}$. It follows that $|U|\le r(r-1)^{n-1}$ and therefore $p_{r,n} \le \frac{r(r-1)^{n-1}}{r^n} = \left(\frac{r-1}{r}\right)^{n-1}$. 
\end{proof}

We also know from \cite{bean1979} that all unavoidable patterns over $[r]$ have length less than $2^n$. Combined with Lemma \ref{thm:avdense} above, we can now obtain an upper bound on the number of unavoidable patterns over $[r]$, where $r > 2$. 

\newtheorem{avupper}[thm]{Proposition}
\begin{avupper}\label{thm:avupper}
Let $r > 2$. The number of unavoidable patters over $[r]$ is at most $r\left(\frac{(r-1)^{2^r-1}-1}{r-2}\right)$.
\end{avupper}
\begin{proof}
 The number of unavoidable patterns of length $n$ is bounded from above by 
$$
p_{r,n}r^n \le \left(\frac{r-1}{r}\right)^{n-1}r^n = r(r-1)^{n-1}.
$$
Since there are no unavoidable patterns of length greater than $2^n-1$ we have the total number of unavoidable patterns is at most 
$$
\sum_{i=1}^{2^r-1}r(r-1)^{i-1} = r \sum_{i=0}^{2^r-2}(r-1)^{i} = r\left(\frac{(r-1)^{2^r-1}-1}{r-2}\right)
$$
and the proof is complete.
\end{proof}

%
%


%
%

\section{Free Letters and Computation}\label{sfree}


We now proceed to investigate the computational aspects of unavoidability, 
assuming a basic familiarity with algorithms and
computational complexity, for which
Hopcroft and Ullman \cite{hop79}  and \cite{cor89} provide authoritative references. The computational complexity of patterns has been the subject of significant study. Rytter and Shur \cite{ryt2014} demonstrated that the problem of finding whether a pattern is reflected in a given string is $NP$-complete. In the same article, they mention that the problem of determining whether a pattern is unavoidable has, at face value, properties that many other $NP$-complete problems have. Below, we show that their suspicions are correct. The complexity of unavoidable words in the sense of substrings, not morphisms, has also been investigated \cite{sadri2009}.


For a pattern $\al$ we construct a directed bipartite graph $\Ga$, which we call the {\sl graph of $\al$}. The vertex set $V(\Ga)$ of $\Ga$ has two nodes $^0ab$ and $^1ab$ for each 2-factor $ab$ of $\al$. The pair of 2-factors $(^0ab, ^1cd)$ of $\al$ is an edge of $\Ga$ whenever $b = d$. Similarly, the pair $(^1ab, ^0cd)$ of $\al$ is an edge of $\Ga$ whenever $a = c$. The reason why we create two vertices for each 2-factor is to prevent paths of the form $xa, xb, xc$.

\newtheorem{gfree}[thm]{Lemma}
\begin{gfree}\label{thm:gfree}
Let $\al$ be a pattern. A letter $x$ of $\al$ is not free if and only if there is a path in $\Ga$ from a node having $x$ as its first component to a node having $x$ as its second component.
\end{gfree}
\begin{proof}
If $x$ is not free for $\alpha$, then there is an $n\in\N$ and $a_1, a_2,\dots,a_n,b_1,b_2,\dots,b_n$ such that 
 \begin{align*}
 xa&_1\\
 b_1&a_1\\
 b_1&a_2\\
 b_2&a_2\\
 &\vdots\\
 b_n&x\\
 \end{align*}
 are all factors of $w$. It is clear from the definition of $\Ga$ that the edges 
  \begin{align*}
 (^0xa_1, ^1&b_1a_1)\\
 (^1b_1a_1&, ^0b_1a_2) \\
 (^0b_1a_2&, ^1b_2a_2)\\
 (^1b_2a_2&, ^0b_2a_3)\\
 &\vdots\\
 (^0b_na&_{n-1},^1b_nx)\\
 \end{align*}
 all exist in $\Ga$. Therefore a path from $^0xa_1$ to $^1b_nx$ exists in $\Ga$, as desired.

Proving the converse is essentially the same as reading the construction above in reverse.
\end{proof}

Given Lemma \ref{thm:gfree}, we can easily construct an efficient algorithm that decides, given a pattern $\al$ and a letter $x$ appearing in $\al$, whether $x$ is free for $\al$. 

Firstly, the construction of the adjacency matrix of $\Ga$ from $\al$ can be defined as follows:

\hfill\break
\par\noindent\rule{\textwidth}{0.4pt}
\begin{verbatim}
def BUILD_G(n, alpha):
    G = [[0 for x in range(2*n-2)] for y in range(2*n-2)]
    V =  [[[0 for x in range(2)] for y in range(2)] for z in range(n-1)]
    for i in range(n-1):
        V[i][0][0] = V[i][1][0] = alpha[i]
        V[i][0][1] = V[i][1][1] = alpha[i+1]
    for i in range(n-1):
        for j in range(n-1):
           if V[i][0][0] == V[j][1][0]:
               G[i + n - 1][j] = 1
           if V[i][0][1] == V[j][1][1]:
               G[i][j + n - 1] = 1
    return (V, G)
\end{verbatim}
\par\noindent\rule{\textwidth}{0.4pt}
\hfill\break

 We notice that the runtime is dominated by the nested for loop and therefore requires $O(n^2)$ computational steps. The subroutine as it is written is not quite optimal since multiple vertices are created if the same 2-factor is repeated. This impacts the time complexity only to a multiplicative constant and simplifies the description.

Now, given a graph $G = \Ga$ and a letter $x$, we can use a standard depth-first search algorithm to detect if $x$ is not free.
For simplicity we write a standard depth-first search subroutine. 
\hfill\break
\par\noindent\rule{\textwidth}{0.4pt}
\begin{verbatim}
def DFS(n, G, V, i, p, x, is_seen):
    is_seen[i][p] = True
    if p == 0:
        q = 1
    else:
        q = 0
    for j in range(n-1):
        if not is_seen[j][q] and 
          (
          (q == 0 and G[i + n - 1][j] == 1) or
          (q == 1 and G[i][j  + n - 1]  == 1)
           ):
            if V[j][q][1] == x:
                return True
            else:
                if DFS(n, G, V, j, q, x, is_seen):
                   return True
    return False
\end{verbatim}
\par\noindent\rule{\textwidth}{0.4pt}
\hfill\break

We are now ready to write the subroutine determining if $x$ is free for $\al$.
\hfill\break
\par\noindent\rule{\textwidth}{0.4pt}
\begin{verbatim}
def IS_FREE(alpha, x):
    if x not in alpha:
       return False
    n = len(alpha)
    V,G = BUILD_G(n, alpha)
    is_seen = [[False for i in range(n)] for j in range(n)]
    for i in range(n-1):
        is_seen = [[False for k in range(n)] for j in range(n)]    
        if V[i][0][0] == x:
            if not is_seen[i][0]:
                if DFS(n, G, V, i, 0, x, is_seen):
                    return(False)
    return(True)
\end{verbatim}
\par\noindent\rule{\textwidth}{0.4pt}
\hfill\break
The subroutine {\tt IS\_FREE} requires $O(n^2)$ computational steps, where $n = |\al|$: We already know that {\tt BUILD\_G} is $O(n^2)$. In subsequent steps, {\tt DFS} is called at most $n$ times since every vertex is marked as seen subsequent to the invocation of {\tt DFS}. At every invocation of {\tt DFS}, at most $n$ neighbors of a vertex are examined.

Let us pause for a moment to remember where we started and what we have seen along the way. Our initial definition of unavoidability sounds distinctly non-finitary: A pattern must be reflected by all but finitely many elements for every set over any finite alphabet. Theorem \ref{thm:bembetter} then gives us a finitary characterization of unavoidability in that we only need to look for a sequence of deletions of free letters. Most recently we have seen, in addition, that the problem of deciding whether a letter is free falls in Polynomial Time. It is hence starting to look as though the problem of determining whether a pattern is free might fall in $NP$: We can nondeterministically guess the sequence of deletions and verify the validity of the guess (each deletion being of a free  letter) in polynomial time. We may also ask how hard this problem is, relative to other problems in $NP$. In the following two sections, we explore this. 

\section{Unavoidability and Logic}\label{slogic}
We work to establish a natural correspondence between boolean formulas and patterns. In particular, we show that given a boolean formula, we can construct a word whose unavoidability coincides with the satisfiability of of the formula. We will restrict our construction to 3-CNF boolean formulas, as the correspondence between this subset of boolean formulas and the set of all boolean formulas is well-understood (see \cite{hop79}).

Let $\phi$ be any 3-CNF boolean formula. We construct $\alp$, the {\sl word of $\phi$}, as follows:
Suppose $\phi$ has $n$ variables $x_1, x_2, \dots, x_n$. Without loss of generality  $\phi = C_1 \land C_2 \land \dots\land C_m$, where each $C_i$ is a clause of the form $(p_{i1}\lor p_{i2}\lor p_{i3})$, each $p_{ij}$ being either a variable $x_k$, or its negation $\ov{x}_k$. We may also assume that any negated variables occur after any non-negated variables in each clause. We start by defining the letters in $\alp$. These letters will fall into the following four categories:

\begin{enumerate}
\item The set $X_{\alp} = \{x_i,\ov{x}_i,: i \le n\} $
\item The set  $Y_{\alp} = \{a_j, b_j, c_j,d_j : j < m\}$
\item The letter $e$
\item The set $Z = \{z_{i}: i \le M\}$. We choose $M$ to be sufficiently large so that every element of this set will appear exactly once in $\al$.
\end{enumerate}

The elements of $Z$ above are used as``separator" letters to prevent unfortunate 2-factors from occurring. We adopt the convention that we will use each letter in $Z$ once and denote each occurrence of a letter from $Z$ in $\alp$ by $\zp$.  We denote the union of the sets of letters itemized above by $\A$.

For each variable $x_i$, we create the factor
\begin{align*}
ex_i \ov{x}_ie\zp\\
\end{align*}

For each clause $C_j$ in $\phi$ we construct a factor $\delta_j$ as the concatenation of the following factors.

Let $x, y$ and $z$ be variables in $\phi$.
If $C_j$ is of the form $x\lor y\lor z$ we add the following factors to $\al$:

\begin{align*}
a_jx\zp\\
b_jx\zp\\
b_jy\zp\\
c_jy\zp\\
c_jz\zp\\
d_jz\zp\\
d_ja_j\zp\\
\\
a_jb_j\zp\\
a_jc_j\zp\\
a_jd_j\zp\\
\\
a_je\zp
\end{align*}

If $C_j$ is of the form $x\lor y\lor \ov{z}$ we add the following factors to $\al$:

\begin{align*}
a_jx\zp\\
b_jx\zp\\
b_jy\zp\\
c_jy\zp\\
c_jd_j\zp\\
\ov{z}d_j\zp\\
\ov{z}a_j\zp\\
\\
a_jb_j\zp\\
a_jc_j\zp\\
d_ja_j\zp\\
\\
a_je\zp
\end{align*}

If $C_j$ is of the form $x\lor \ov{y}\lor \ov{z}$ we add the following factors to $\al$:

\begin{align*}
a_jx\zp\\
b_jx\zp\\
b_jc_j\zp\\
\ov{y}c_j\zp\\
\ov{y}d_j\zp\\
\ov{z}d_j\zp\\
\ov{z}a_j\zp\\
\\
a_jb_j\zp\\
c_ja_j\zp\\
d_ja_j\zp\\
\\
a_je\zp
\end{align*}

If $C_j$ is of the form $\ov{x}\lor \ov{y}\lor \ov{z}$ we add the following factors to $\al$:

\begin{align*}
a_jb_j\zp\\
\ov{x}b_j\zp\\
\ov{x}c_j\zp\\
\ov{y}c_j\zp\\
\ov{y}d_j\zp\\
\ov{z}d_j\zp\\
\ov{z}a_j\zp\\
\\
b_ja_j\zp\\
c_ja_j\zp\\
d_ja_j\zp\\
\\
ea_j\zp
\end{align*}

We define the word $\alp$ of $\phi$ as the culmination of the above construction and proceed to prove some properties of $\alp$.

\newtheorem{ynotfree}[thm]{Lemma}
\begin{ynotfree}\label{thm:ynotfree}
Let $\phi = C_1 \land C_2 \land \dots\land C_m$ be a 3-CNF boolean formula. Let $B\subset \A\setminus \{a_j,b_j,c_j,d_j\}$ be such that, if $p_i$ is a literal in $C_j$, then the letter $p_i$ is not in $B$.
No  letter in $\{a_j, b_j, c_j,d_j\}$ is free for $\alp^B$.
\end{ynotfree}
\begin{proof}
If $C_j$ is of the form $x\lor y\lor z$ then the path

\begin{align*}
&a_jx\\
&b_jx\\
&b_jy\\
&c_jy\\
&c_jz\\
&d_jz\\
&d_ja_j
\end{align*}
shows $a_j$ is not free. Similarly the path
\begin{align*}
&b_jx\\
&a_jx\\
&a_jb_j
\end{align*}
yields that $b_j$ is not free, while 
\begin{align*}
&c_jy\\
&b_jy\\
&b_jx\\
&a_jx\\
&a_jc_j\\
\end{align*}
and 
\begin{align*}
&d_jz\\
&c_jz\\
&c_jy\\
&b_jy\\
&b_jx\\
&a_jx\\
&a_jd_j\\
\end{align*}
demonstrate that $c_j$ and $d_j$ are not free.
The arguments for the remaining three cases where $C_j$ contains negated variables are substantially similar.
\end{proof}

The following lemma is easily established by inspecting $\alp$.
\newtheorem{bfreed}[thm]{Lemma}
\begin{bfreed}\label{thm:bfreed}
Let $\phi = C_1 \land C_2 \land \dots\land C_m$ be a 3-CNF boolean formula and let $\alp$ be the word of $\phi$. For $i \le m$ the letters $b_j, c_j$ and $d_j$ are free for $\alp^{a_j}$
\end{bfreed}

\newtheorem{dnotfree}[thm]{Lemma}
\begin{dnotfree}\label{thm:dnotfree}
Let $\phi = C_1 \land C_2 \land \dots\land C_m$ be a 3-CNF boolean formula. Let $B = \{y_1, y_2, \dots, y_k\} \subseteq \A\setminus \{e\}$. Suppose $y_1, y_2, \dots, y_k$ is a free deletion sequence for $\alp$. Suppose furthermore that $B$ is such that, if $p_i$ is a literal in $C_j$, then the letter $p_i$ is not in $B$. Then $e$ is not free for $\alp^B$.
\end{dnotfree}
\begin{proof}
Suppose $C_j$ and $B$ are as in the statement of the Lemma. Since $y_1, y2, \dots,y_k$ is a free deletion sequence, Lemma  \ref{thm:ynotfree} gives us that none of the letters $a_j, b_j, c_j$ and $d_j$ are in $B$. 

If $x_i$ is the first literal in $C_j$, then the path
\begin{align*}
&ex_i\\
&a_jx_i\\
&a_je\\
\end{align*}
ensures that $e$ is not free. 

On the other hand, if $\ov{x_i}$ is the first literal in $C_j$, then we know (using our assumption that negated variables always appear after non-negated variables in a clause) that $C_j$ is of the form $\ov{x_i}\lor \ov{y}\lor \ov{z}$, where $y$ and $z$ are variables of $\phi$ and the path
\begin{align*}
&ea_j\\
&\ov{z}a_j\\
&\ov{z}e\\
\end{align*}
shows that $e$ is not free. 
\end{proof}

\newtheorem{xxbar}[thm]{Lemma}
\begin{xxbar}\label{thm:xxbar}
Let $B\subseteq \A\setminus\{e\}$. If there is an $i$ such that both $x_i$ and $\ov{x}_i$ are in $B$, then $\alp^B$ is avoidable.
\end{xxbar}
\begin{proof}
If $x_i$ and $\ov{x}_i$ are both in $B$, then $ex_i\ov{x}_ie^B = ee$ is a factor of $\alp^B$.
\end{proof}

\newtheorem{zfree}[thm]{Lemma}
\begin{zfree}\label{thm:zfree}
Let $w$ be a word of the form $\zp\zp\dots\zp$. Every letter in $w$ is free.
\end{zfree}
\begin{proof}
Each of the letters $\zp$ appears at most once in $w$.
\end{proof}


\newtheorem{xfree}[thm]{Lemma}
\begin{xfree}\label{thm:xfree}
Let $\phi = C_1 \land C_2 \land \dots\land C_m$ be a 3-CNF boolean formula in $n$ variables and let $\al = \al_{\phi}$  be the word of $\phi$. Fix $k < n$. Let $S_k = \{p_1, p_2, \dots, p_k\}$, where for each $i$, either $p_i = x_i$ or $p_i = \ov{x}_i$. Both $x_{k+1}$ and $\ov{x}_{k+1}$ are free for $\alp^{S_k}$.
\end{xfree}
\begin{proof}
We proceed by induction on $k$. For $k = 0$, we have $S_k = \emptyset$.
Hence, the only 2-factor (excluding those containing $\zp$) that contains $x_1$ as the first letter is $x_1\ov{x}_1$ and the only 2-factor that contains $\ov{x}_1$ as the second letter is $x_1\ov{x}_1$. So the only path starting at a 2-factor having $x_1$ as the first first letter is the one-cycle from $x_1\ov{x}_1$ to itself. Our base case has thus been established.

Now suppose the lemma holds for some $k$. Again, the only 2-factor containing $x_k$ as the first letter is $x_k\ov{x}_k$ and the only 2-factor that contains $\ov{x}_k$ as the second letter is $x_k\ov{x}_k$. The lemma immediately follows.
\end{proof}


\newtheorem{zlast}[thm]{Lemma}
\begin{zlast}\label{thm:zlast} Let $\phi$ be a 3-CNF boolean formula and let $\al = \al_{\phi}$  be the word of $\phi$.  If $\al$ is unavoidable, then there is a free deletion sequence where the letters $\zp$ are deleted after all other letters are deleted.
\end{zlast}
\begin{proof}
By Lemma \ref{thm:zfree}, it suffices to note that deleting any letter $\zp$ cannot make any letter free that is not already free.
\end{proof}

\newtheorem{delsingle}[thm]{Lemma}
\begin{delsingle}\label{thm:delsingle} Let $\phi= C_1 \land C_2 \land \dots\land C_m$ be a 3-CNF boolean formula and let $\al = \al_{\phi}$  be the word of $\phi$.  If $\al$ is unavoidable, then there is a free deletion sequence where every free set that is deleted contains exactly one letter.
\end{delsingle}
\begin{proof}
Suppose $\al$ is as in the statement of the lemma. There is a partition of $\A$ into sets $B_1, B_2, \dots, B_k$ such that, for every $i<k$, we have that $B_{i+1}$ is a free set for $\al^{B_1,B_2,\dots,B_i}$. Assume for contradiction that there is some $t$ such that $|B_t| > 1$ and every deletion sequence of the individual letters in $B_t$ results in no letter $y\in \A\setminus D\setminus X_{al}$ being free for $\al^D$, where $D = B_1\cup B_2\cup\dots\cup B_{t-1}\cup E$ and $E\subset B_t$ is the set of letters in $B_t$ that are already deleted.  Let $x$ be the last letter in $E$ that was deleted and let $y\in B_t\setminus E$ be not free for $\al^D$.

{\raggedright \sl Case 1. } $x = x_i$ for some $i$. 

{\sl Subcase 1.1. } $y = x_j$ for some $j$. We have that the deletion of $x_i$ resulted in a path from a 2-factor having $x_j$ as its first component to a 2-factor that has $x_j$ as its second component. We observe that the only only 2-factors that can possibly be newly created by the deletion of $x_i$ are among the following forms: 
\begin{enumerate}
\item $a_k\zp$, $b_k\zp$, $c_k\zp$ or $d_k\zp$. Since each letter $\zp$ appears only once in $\al$ we can conclude that these factors are not in our path.
\item $e\ov{x}_i$. We conclude that there is a path from a 2-factor having $x_j$ as its first letter to $e\ov{x}_i$. But this means there is a path in $\al^{D\setminus \{x_i\}}$ from a 2-factor having $x_j$ as its first letter to $ex_i$. Thus $x_i$ and $x_j$ cannot be in the same free set, a contradiction.
\item $ee$. Similarly to the previous item, we conclude that $x_i$ and $x_j$ cannot be in the same free set as it implies a path from $x_ie$ to a 2-factor having $x_j$ as its second component before the deletion of $x_i$.
\end{enumerate}

{\sl Subcase 1.2. } $y = a_j$, or $y = b_j$, or $y = c_j$, or $y = d_j$ for some $j$. We follow the same reasoning as Subcase 1.1 and arrive at the same conclusion, showing $y$ not free implies either a path from a 2-factor having $y$ as its first letter to a two factor having $x$ as its second component, or vice versa.

{\sl Subcase 1.3. } $y = e$. Our reasoning is substantially similar to the previous two subcases.

{\raggedright \sl Case 2. } $x = \ov{x}_j$ for some $j$. This is symmetric to Case 1.

{\raggedright \sl Case 3. }  $y = a_j$,  $y = b_j$, $y = c_j$, or $y = d_j$ for some $j$. The deletion of $x$ results only in new 2-factors containing one or more of the $\zp$ letters, so a new path from a 2-factor having $y$ as its first letter to a 2-factor having $y$ as its second letter could not have been created by virtue of deleting $x$.

{\raggedright \sl Case 4. } $x = e$.  We know $y \ne x_j$ for any $j$ since this would imply the existence of the 2-factor $ex_j$, negating the assumption that $x_j$ and $e$ are in the same free set. Similarly $y\ne \ov{x}_j$ by virtue of the 2-factor $\ov{x}_je$ and $y\ne a_j$ 
because of either $a_je$ or $ea_j$ would have been a 2-factor before the deletion of $e$, so we are left with the possibilities of $y = b_j$, $y = c_j$ or $y = d_j$.  

Suppose $y = b_j$. From Lemma \ref{thm:bfreed} we have that $a_j\notin D$.
 If $C_j$ consists of three negated variables, then we know that $\ov{z}\in D$, where $\ov{z}$ is the last literal in $C_j$, for otherwise the path 
\begin{align*}
&ea_j\\
&\ov{z}a_j\\
&\ov{z}e\\
\end{align*}
would contradict the assumption that $e$ is free. But then there is no path from a 2-factor having $b_j$ as its first component to a 2-factor having $b_j$ as its second component, contradicting that $b_j$ is not free. On the other hand, if $C_j$ contains a non-negated variable, we arrive at a similar contradiction using the path
 \begin{align*}
&ez\\
&a_jz\\
&a_je\\
\end{align*}
where $z$ is the first literal in $C_j$.
The arguments for $y = c_j$ and $y = d_j$ substantially identical.
The cases are exhausted. This concludes the proof.
\end{proof}

\newtheorem{delorder}[thm]{Lemma}
\begin{delorder}\label{thm:delorder} Let $\phi = C_1 \land C_2 \land \dots\land C_m$ be a 3-CNF boolean formula in $n$ variables and let $\al = \al_{\phi}$  be the word of $\phi$.  If $\al$ is unavoidable, then there is a deletion sequence of free letters that starts by deleting either $x_i$ or $\ov{x}_i$, for $i \le n$.
\end{delorder}
\begin{proof}
Suppose $\al$ is unavoidable. By Lemma \ref{thm:delsingle} there is a deletion sequence of free letters reducing $\al$ to the empty word. We may assume by Lemma \ref{thm:zlast} that all the letters $\zp$ appear at the end of the deletion sequence.
We know from Lemma \ref{thm:xfree} that it is possible to delete $x_i$ or $\ov{x}_i$ as the $i$th letter in a deletion sequence of free letters. We need to establish that we can alter any deletion sequence of free letters to one where the first $n$ deletions are as described by Lemma \ref{thm:xfree}.

It suffices to show that we can always invert the deletion order whenever an $x\in X_{\al}$ is deleted immediately after some letter $y \notin X_{\al}$ and $\ov{x}$ is deleted after $x$, where $\ov{\ov{x}} = x$. 

{\raggedright\sl Case 1.} $y = a_j$. We start by noting that $x$ is already free before the deletion of $a_j$ since no new 2-factor that does not contain a letter $\zp$ is created by deleting any letter not in $X_{\al}$ or in $Z$. Suppose $x$ is a non-negated variable, i.e. $x = x_i$ for some $i\le n$. Suppose for contradiction that inverting the deletion order of $x_i$ and $a_j$ results in $a_j$ not being free. We notice that  the only new 2-factor (excluding ones with $\zp$ letters) created by the deletion of $x_i$ is $e\ov{x}_i$, so after the deletion of $x_i$ there is a path from a 2-factor having $a_j$ as its first letter to $e\ov{x}_i$ and a path from $e\ov{x}_i$ to a 2-factor having $a_j$ as its second letter. We now notice that the only 2-factor having $\ov{x}_i$ as its second letter is $e\ov{x}_i$, so the immediate predecessor to $e\ov{x}_i$ in our malignant path has $e$ as its first letter. But this means that the immediate successor to $e\ov{x}_i$ in the path has $\ov{x}_i$ as its second letter. But again the only 2-factor having $\ov{x}_i$ as its second letter is $e\ov{x}_i$, so the path cannot proceed to any 2-factor not already in the path. Hence there is already a 2-factor having $a_j$ as its second letter at some earlier point in the path, contradicting our assumption that $a_j$ was free before $x_i$ was deleted.  Supposing, on the other hand,  that $x = \ov{x}_i$ leads to the same contradiction through symmetric reasoning, where we end up in a dead end at the 2-factor $xe$.

{\raggedright\sl Case 2.} $y \in \{b_j, c_j, d_j\}$.  The argument is essentially the same as Case 1.

{\raggedright\sl Case 3.} $y = e$. Suppose again, for contradiction, that $e$ is not free as the result of deleting $x_i$. Again the only new 2-factor created is $e\ov{x}_i$, so there is a path from $e\ov{x}_i$ to a 2-factor having $e$ as its second letter. But since the only 2-factor having $\ov{x}_i$ as its second letter is $e\ov{x}_i$, we find ourselves back at the contradiction described in Case 1. For $x = \ov{x}_i$ the argument is, once again, symmetric.

The cases are exhausted and the proof is complete.
\end{proof}

\newtheorem{satiff}[thm]{Proposition}
\begin{satiff}\label{thm:satiff} 
If $\phi$ is a 3-CNF boolean formula and $\al = \alp$ is the word of $\phi$, then $\phi$ is satisfiable if and only of $\al$ is unavoidable. 
\end{satiff}
\begin{proof}
Suppose $\phi$ with variables $x_1, \dots, x_n$ and clauses $C_1, \dots, C_m$ is satisfiable. 
Let $x_1 = e_1, x_2 = e_2, \dots, x_n = e_n$, with each $e_i\in\{0,1\}$, be a satisfying assignment for $\phi$. We show that $\alp$ will reduce to the empty set by deleting all its letters in the following stages:
\begin{enumerate}
\item For $i \le n$, delete $x_{i}$ if $e_i = 1$, otherwise delete $\ov{x}_{i}$. 
\item Next, for $j \le m$, delete $a_j$, $b_j$, $c_j$ and the $d_j$. 
\item Delete the letter $e$.
\item Delete the remaining $x_{i}$ and $\ov{x}_{i}$.
\item Delete the remaining characters $\zp$ in any order.
\end{enumerate}
Furthermore, every letter that is deleted will be free at the stage when the deletion happens.

Lemma \ref{thm:xfree} guarantees that every deletion in Stage (1) above is of a free letter. Since $\phi$ is satisfiable, every clause $C_j = (p_{1}\lor p_{2}\lor p_{3})$ has at least one literal that is set to $1$. If $p_1 = x_i = 1$, then $x_i$ is deleted in Stage (1). Consequently $a_j$ is free after Stage (1) and can be deleted in Stage (2). The deletion of $a_j$, in turn, causes $b_j$, $c_j$ and $d_j$ to become free. The remaining cases among $p_k = {x}_i = 1$ and $p_k = \ov{x}_i = 0$ lead to $a_j$, $b_j$, $c_j$ and $d_j$ being deleted in a similar fashion. We can therefore successfully complete the deletions in Stage (2).

After the completion of Stage (2) the only 2-factors (once again ignoring the $\zp$) containing $e$, are of the form $ep_i$ and $p_ie$, where for each $i$ we have either $p_i = x_i$ or $p_i = \ov{x}_i$. Furthermore, for each $i$ the same 2-factors are the only ones containing $p_i$. Therefore $e$ is free and consequently Stage (3) can be completed. 

After the completion of Stage (3), there are no 2-factors left that do not contain one of the $\zp$. Since every letter $\zp$ is unique, we can safely complete Stage (4). Now all that remains is letters of the form $\zp$ and hence, using Lemma \ref{thm:zfree}, we can delete the remaining letters. It follows, by Theorem \ref{thm:bembetter}, that $\alp$ is unavoidable, as desired.

Now suppose $\phi$ is unsatisfiable. For contradiction, suppose $\alp$ is unavoidable. Using Lemma \ref{thm:delorder}, we may assume that the first $n$ deletions are $p_1, p_2, \dots, p_n$ with, for every $i$,  either $p_i = x_i$ or $p_i = \ov{x}_i$. Define the following assignment on $\phi$: If $p_i = x_i$, then set the variable $x_i$ to $1$, otherwise set $x_i$ to $0$. Since $\phi$ is not satisfiable, we know that there is some clause $C_j$ that is not satisfied by our chosen assignment. But this means that none of the $p_i$ in the first $n$ deletions appear in $C_j$ and consequently none of the letters $a_j$, $b_j$, $c_j$ and $d_j$ are free after the first $n$ deletions, by Lemma \ref{thm:ynotfree}. In addition, by Lemma \ref{thm:dnotfree}, we have that $e$ is not free. In order to free any of these letters, we have to delete at least one letter $x_i$ or $\ov{x}_i$ which has, thus far not been deleted. But this means, for some $i$, both $x_i$ and $\ov{x}_i$ have been deleted. Using Lemma \ref{thm:xxbar}, we have a contradiction.
\end{proof}

%
%

\section{Unavoidability and Computational Complexity}\label{scomplexity}

We define the {\sl Word Unavoidability Problem} as follows: Given a pattern $\al$ over a finite alphabet, determine if $\al$ is unavoidable. We refer to the set of unavoidable patterns as $WU$.

\newtheorem{unnpc}[thm]{Theorem}
\begin{unnpc}\label{thm:unnpc}
The Word Unavoidability Problem is $NP$-complete.
\end{unnpc}
\begin{proof}
We note that, given a 3-CNF boolean formula $\phi$, the construction of the word $\alp$ of $\phi$ requires a number of computational steps that is linear in the length of $\phi$: For every variable $x_i$, we need to add a factor $dx_i\ov{x}_id$. For every clause we need to add a constant number of factors that are derived purely from the literals in that clause.

Proposition \ref{thm:satiff} therefore leaves us very little work to do. All that remains is to prove $WU\in NP$. 
Using Theorem \ref{thm:bemmain} and the algorithm {\tt IS\_FREE} above, we write the following test for unavoidability:
\hfill\break
\par\noindent\rule{\textwidth}{0.4pt}

\begin{verbatim}
IS_UNAVOIDABLE(alpha)
    A[] = the distinct letters in alpha
    B[] = the distinct letters in alpha and all pairs of letters in A
    n = |B|
    nondeterministically guess the permutation pi on [n]
        for i = 1 to n:
            if B[pi(i)] is a single letter and occurs in alpha:
               x = B[pi(i)]
               if IS_FREE(alpha, x):
                    delete every occurrence of x from alpha
               else:
                    nondeterministic guess dies
            else:
               x, y = B[pi(i)]
               if x and y are both letters in alpha:
                  replace every occurrence of y in alpha with x
        return True  
     return False             
\end{verbatim}
\par\noindent\rule{\textwidth}{0.4pt}
\hfill\break

Each branch of nondeterminism completes at most $|A_{\al}|$ deletions and $|A_{\al}|^2$ identifications of letters.  Since {\tt IS\_FREE} runs in polynomial time, so does each branch of {\tt IS\_UNAVOIDABLE}. The number of branches of nondeterminism is bounded from above by the number of permutations on $|A_{\al}| + |A_{\al}|^2$.
\end{proof}

%
%

\section{Conclusion}

Many interesting questions remain regarding the complexity of unavoidable patterns \cite{cur93}. The bounds established in Theorem \ref{thm:mainres} above are not primitive recursive. We do not know if there is a primitive recursive upper bound, nor do we know what lower bounds exist, for any significantly general subset of patterns.

\section{Acknowledgements}
This article has been written in partial fulfillment  of the requirements for the degree Doctor of Philosophy in Operations Research at the University of South Africa. Special and sincere thanks go to Willem Fouch\'e and Petrus Potgieter for continuous insight and guidance. Many thanks also to Narad Rampersad and James Currie who read earlier versions of this paper and communicated problems to me. 

%
%
%
\bibliographystyle{jloganal}
\bibliography{cs}{}
%



\end{document}